\documentclass{article} 

\usepackage{graphicx} 
\usepackage[margin=1in]{geometry}
\usepackage{algpseudocode}
\usepackage{algorithm, comment}
\usepackage{amssymb}
\usepackage{amsthm}
\usepackage{amsmath}
\usepackage{bbm}
\newtheorem{theorem}{Theorem}
\newtheorem{lemma}{Lemma}

\title{Column Bound for Orthogonal Matrix Factorization}
\author{
  \begin{tabular}{cc}
    Anirudh Dash \\
    \small Department of Electrical Engineering  \\
    \small Indian Institute of Technology, Hyderabad \\
    \small Hyderabad, Telangana, India \\
    \small \texttt{ee21btech11002@iith.ac.in} 
  \end{tabular}
}
\date{}

\newcommand{\keywords}[1]{\textbf{Keywords:} #1}

\begin{document}

\maketitle

\begin{abstract}
    This article explores the intersection of the Coupon Collector's Problem and the Orthogonal Matrix Factorization (OMF) problem. Specifically, we derive bounds on the minimum number of columns \(p\) (in \(\mathbf{X}\)) required for the OMF problem to be tractable, using insights from the Coupon Collector's Problem. Specifically, we establish a theorem outlining
    the relationship between the sparsity of the matrix \(\mathbf{X}\) and the number of columns \(p\) required to recover the matrices \(\mathbf{V}\) and \(\mathbf{X}\) in the OMF problem. We show that the minimum number of columns \(p\) required is given by \(p = \Omega \left(\max \left\{ \frac{n}{1 - (1 - \theta)^n}, \frac{1}{\theta} \log n \right\} \right)\), where \(\theta\) is the i.i.d Bernoulli parameter from which the sparsity model of the matrix \(\mathbf{X}\) is derived.
\end{abstract}


\section{Introduction}

The Coupon Collector's Problem \cite{mitzenmacher2005probability} is a well-studied problem in probability theory with applications in various fields such as computer science, statistics, and information theory. The problem is as follows: Suppose there are \(n\) different types of coupons, and each time a coupon is collected, it is independent and uniformly distributed among the \(n\) types. The problem is to find the expected number of coupons that need to be collected to obtain at least one of each type.
This classic problem has been analyzed extensively, and it is known that the expected number of coupons needed is given by \(n H_n\), where \(H_n\) is the \(n\)-th harmonic number, approximately \(\log n + \gamma\) (where \(\gamma\) is the Euler-Mascheroni constant and \(\log\) represents the natural logarithm) for large \(n\) \cite{mitzenmacher2005probability}.

The Orthogonal Matrix Factorization (OMF) problem is posed as follows: Given a matrix \(\mathbf{Y} \in \mathbb{R}^{n \times p}\), we wish to find matrices \(\mathbf{V} \in \mathbb{R}^{n \times n}\) and \(\mathbf{X} \in \mathbb{R}^{n \times p}\) such that \(\mathbf{Y} = \mathbf{V}\mathbf{X}\), where \(\mathbf{V}\) is an orthogonal matrix. To find the matrices \(\mathbf{V}\) and \(\mathbf{X}\), we need to solve the optimization problem:
\begin{equation}
    \min_{\mathbf{V}, \mathbf{X}} \|\mathbf{Y} - \mathbf{V}\mathbf{X}\|_F^2
\end{equation}
where \(\|\cdot\|_F\) denotes the Frobenius norm. The problem is non-convex and NP-hard. However, under certain conditions on the matrix \(\mathbf{Y}\), it is possible to find the matrices \(\mathbf{V}\) and \(\mathbf{X}\) using various algorithms. Here, we analyze the minimum number of columns \(p\) required for the above problem to be tractable.
We derive the bound on \(p\) using the Coupon Collector's Problem and the fact that to recover \(\mathbf{V}\) and \(\mathbf{X}\), there must exist
at least one linear combination of each of the columns of \(\mathbf{V}\) \cite{spielman2012exact}.

\section{Theorem}

\begin{theorem}
    Consider the matrices \(\mathbf{Y} \in \mathbb{R}^{n \times p}\), \(\mathbf{V} \in \mathbb{R}^{n \times n}\), \(\mathbf{X} \in \mathbb{R}^{n \times p}\) such that \(\mathbf{Y} = \mathbf{V}\mathbf{X}\), where \(\mathbf{V}\) is an orthogonal matrix. Then, if the value of \(p\) satisfies:
    \begin{equation}
        p = \Omega \left(\max \left\{ \frac{n}{1 - (1 - \theta)^n}, \frac{1}{\theta} \log n \right\} \right)
    \end{equation}
    all the rows of \(\mathbf{X}\) are non-zero with high probability, under the following assumptions:
    \begin{enumerate}
        \item The sparsity model of \(\mathbf{X}\) is drawn from an independent Bernoulli distribution with parameter \(\theta\), where \(\theta \in (0,1]\). Thus, for each \(\mathbf{X_{ij}}\), \(1 \leq i \leq n, 1 \leq j \leq p\), we have:
              \[
                  \mathbbm{1_{\mathbf{X_{ij} \neq 0}}} = B_{ij}, \quad B_{ij} \sim B(\theta)
              \]
        \item The entries \(\mathbf{X_{ij}}\) are arbitrary real numbers.
    \end{enumerate}
    \label{thm:1}
\end{theorem}

Note that there is no explicit constraint on the value of \(\theta\). \( B_{ij} \sim B(\theta) \) denotes that an entry of \(\mathbf{X}\) is non-zero with probability \(\theta\) and
zero with probability \(1 - \theta\).

\newpage
\setcounter{section}{0}
\renewcommand{\thesection}{\Alph{section}}

\section{Appendix}
\subsection{Proof of Theorem \ref{thm:1}}

\begin{proof}
    \begin{lemma}
        The expected number of columns required to find at least \(1\) non-zero element in each row of the matrix \(\mathbf{X} \) is given by:
        \begin{equation}
            p = \sum_{k=0}^{n-1} \frac{1}{1 - \sum_{r=0}^{k} \binom{k}{r} (\theta)^r (1-\theta)^{n-r}}
        \end{equation}
    \end{lemma}

    \begin{proof}
        In the Coupon Collector's Problem, we collect coupons independently. There are a total of \(n\) unique coupons. The expression for the expected number of coupons we need to collect to recover all coupons is:
        \begin{equation}
            \text{Expected coupons} (p) = \sum_{k=0}^{n-1} \frac{1}{1 - \frac{k}{n}}
        \end{equation}
        where \(k\) is the number of unique coupons that have already been found. The denominator is \(1 -\) the probability of finding a coupon that has already been found beforehand.
        The above scenario is closely related to our case. With every new column, we are trying to find one or more new row indices (i.e., these row indices have a non-zero entry). In order for us to
        be able to recover the matrices \(\mathbf{V}\) and \(\mathbf{X}\) with high probability, we need to ensure that all the rows of \(\mathbf{X}\) have at least one non-zero element. The difference in our case (from the
        standard Coupon Collector's Problem)
        is that we may find multiple new row indices with each column.

        In our case, say at some point, k unique rows have been identified (i.e., at some point up to the current column, there has been a non-zero element in each of these rows- to model this, define the set $S$ as the set of row indices that have already been identified, i.e.,  $S=\{i: \exists l(<j) \ s.t \ \mathbf{X_{il}}\neq 0 \}$
        and $\lvert S \rvert = k$.)

        The following cases all contribute to the non-uniqueness of row indices found in the next column ($j$):
        \begin{enumerate}
            \item \(\mathbf{X_{ij}}=0 \ \forall i\) (i.e., all elements of the $j^{th}$ column are 0)
            \item \(\mathbf{X_{ij}}\neq 0 \) for exactly $r(\neq 0)$ row indices belonging to $S$. Thus, exactly $r$ of the $k$ already found row indices non-zero, and all the other elemets in the $j^{th}$ column of $\mathbf{X}$ are 0. $r$ can be any number between $1$ and $k$. The number of ways to choose $r$ row indices from $k$ is $\binom{k}{r}$. For all such cases, as long as the $r$ indices are chosen from $S$, new row indices are not being identified.
        \end{enumerate}
        Thus, we have the expression for the expected number of columns required to find at least one non-zero element in each row of the matrix \(\mathbf{X}\).
    \end{proof}

    Now, we find a lower bound for this summation.
    We are trying to minimize the lower summation to maximize the denominator to find a lower bound on the overall summation. Thus, we get:

    $$
        p = \sum_{k=0}^{n-1} {1 \over 1-(1-\theta)^{n-k}\sum_{r=0}^{k}{{k \choose r}(\theta)^r(1-\theta)^{k-r}}}
    $$

    $$
        = \sum_{k=0}^{n-1} {1 \over 1-(1-\theta)^{n-k}}
    $$

    Thus
    $$
        p \geq \sum_{k=0}^{n-1} {1 \over 1-(1-\theta)^{n}}
    $$

    Since each term is independent of \(k\), we can simplify the above expression to:

    \begin{equation}
        p \geq {n \over 1-(1-\theta)^{n}} \label{1}
    \end{equation}

    This is the lowest possible bound irrespective of $\theta$.
    Another method to find the bound is to solve the summation and then find a minimum value (here, we will later assume the value of $\theta$ to be small): \\

    $$
        p=\sum_{k=0}^{n-1} {1 \over 1-(1-\theta)^{n-k}}
    $$

    This is equivalent to:

    $$
        p=\sum_{k=1}^{n} {1 \over 1-(1-\theta)^{k}}
    $$

    $$
        =n-{(\gamma+\psi^{(0)}(n+1)) \over \log (1-\theta)}
    $$

    where $\psi^{(0)}(n+1)$ is the q-digamma function. We can use the minimum value of the q-digamma function \cite{gordon1994stochastic} to get a lower bound on $p$:

    $$
        \psi^{(0)}(n+1)> \log(n+1)-{1 \over 2(n+1)} -{1 \over 12(n+1)^2}
    $$

    since $\log(1-\theta)$ is negative, the second term in the above sum is positive. We then use the minimum value of the q-digamma function. This gives:

    \begin{equation}
        p \geq n-{(\gamma+\log(n+1)-{1 \over 2(n+1)} -{1 \over 12(n+1)^2}) \over \log(1-\theta)}
    \end{equation}

    We can now use the Taylor series expansion for $\log(1+x)$
    $$
        \log(1+x)=x-{x^2 \over 2}+{x^3 \over 3}-{x^4 \over 4}+\cdots
    $$
    Thus,
    $$
        \log(1-\theta)=(-\theta)-{(-\theta)^2 \over 2}+{(-\theta)^3 \over 3}-{(-\theta)^4 \over 4}+\cdots
    $$
    This gives
    $$
        p \geq n-{(\gamma+log(n+1)-{1 \over 2(n+1)} -{1 \over 12(n+1)^2}) \over ((-\theta)-{(-\theta)^2 \over 2}+{(-\theta)^3 \over 3}-{(-\theta)^4 \over 4}+\cdots)}
    $$
    Solving further, we get:
    $$
        p \geq n+{(\gamma+log(n+1)-{1 \over 2(n+1)} -{1 \over 12(n+1)^2}) \over (\theta)(1+{(\theta) \over 2}+{(\theta)^2 \over 3}+{(\theta)^3 \over 4}+\cdots)}
    $$

    1+${(\theta) / 2}+{(\theta)^2 / 3}+{(\theta)^3 / 4}+\cdots \simeq 1$ for small values of $\theta$. Hence,

    $$
        p \geq n+{(\gamma+\log(n+1)-{1 \over 2(n+1)} -{1 \over 12(n+1)^2}) \over (\theta)}
    $$

    Furthermore, asymptotically, we can ignore $-{1 / 2(n+1)} -{1 / 12(n+1)^2}$, which gives us:

    $$
        p \geq n+{(\gamma+log(n+1)) \over (\theta)}
    $$

    Asymptotically, $\forall \theta=O({(\log n) / n})$, we get:

    $$
        p=\Omega \left({1 \over \theta} \log n \right)
    $$

    For $\theta$ as small as $O(1/n)$, we get $p = \Omega(n \log n)$. This matches the bound in \cite{spielman2012exact}.
\end{proof}
\end{document}